\documentclass[12pt]{article}

\usepackage{graphicx}
\usepackage{pgfplots}
\usepackage{booktabs}
\usepackage{amsmath}
\usepackage{algorithm}
\usepackage{algorithmic}
\usepackage{amssymb}
\usepackage{amsfonts}
\usepackage{latexsym}
\usepackage{color}
\usepackage{subfigure}
\usepackage{latexsym}
\usepackage{graphicx,epstopdf}
\usepackage{subfigure}
\usepackage{authblk}
\usepackage{caption}
\usepackage{mathtools}
\usepackage{cancel}
\usepackage[normalem]{ulem}

\newcommand{\be}{\begin{equation}}
\newcommand{\en}{\end{equation}}
\newcommand{\bea}{\begin{eqnarray}}
\newcommand{\ena}{\end{eqnarray}}
\newcommand{\beano}{\begin{eqnarray*}}
\newcommand{\enano}{\end{eqnarray*}}
\newcommand{\bee}{\begin{enumerate}}
\newcommand{\ene}{\end{enumerate}}

\newcommand{\Hil}{{\cal H}}

\newcommand{\F}{{\cal F}}

\newcommand{\Lc}{{\cal L}}
\newcommand{\D}{{\cal D}}

\newcommand{\E}{{\cal E}}

\newcommand{\Sc}{{\cal S}}

\newcommand{\1}{1 \!\! 1}
\newtheorem{thm}{Theorem}

\newtheorem{lemma}[thm]{Lemma}
\newtheorem{prop}[thm]{Proposition}
\newtheorem{defn}[thm]{Definition}

\newtheorem{remark}{Remark}
\newenvironment{proof}{\noindent {\bf Proof:}}{\hfill$\Box$}
\begin{document}

\title{Tridiagonality, supersymmetry and non self-adjoint Hamiltonians}

\author{F. Bagarello$^{1,2}$, F. Gargano$^{1}$, F. Roccati$^3$\\

\small{
$^1$Dipartimento di Ingegneria -  Universit\`a di Palermo,
Viale delle Scienze, I--90128  Palermo, Italy,}\\
\small{ $^2$I.N.F.N -  Sezione di Napoli,}\\
\small{
$^3$Dipartimento di Fisica e Chimica - Emilio Segr\`e, Universit\`a degli Studi di Palermo, 
via Archirafi 36, I-90123 Palermo, Italy.}\\
\small{\emph{Email addresses:}\\
fabio.bagarello@unipa.it,
francesco.gargano@unipa.it,
federico.roccati@unipa.it}
}

\date{}
\maketitle
\begin{abstract}
In this paper we consider some aspects of tridiagonal, non self-adjoint, Hamiltonians and of their supersymmetric counterparts. In particular, the problem of factorization is discussed, and it is shown how the analysis of the eigenstates of these Hamiltonians produce interesting recursion formulas giving rise to biorthogonal families of vectors. Some examples are proposed, and a connection with bi-squeezed states is analyzed.
\end{abstract}


\section{Introduction}\label{sec:intro}

Few years ago  some authors have discussed tridiagonal Hamiltonians, and their factorization, in connection with Supersymmetric quantum mechanics (SUSY QM) and with an eye to orthogonal polynomials, \cite{yam}. Their idea was to show how certain self-adjoint (infinite) tridiagonal matrices can be written as product of two operators, and how these operators can also be used to deduce results on the Susy partner of the original matrix. The construction the authors propose give rise to a three-terms recurrence relation which they analyse in connection with orthogonal polynomials. These polynomials are constructed both for $H=X^\dagger X$, and for its Susy counterpart $H_{susy}=XX^\dagger$.  

In this paper we extend the analysis in the context of tridiagonal matrices which are \underline{not} necessarily self-adjoint. In particular, we do not assume that the diagonal elements are real, and that the non zero off-diagonal entries are related by any symmetry. The rationale behind this choice is that, as we will discuss in Section \ref{sect3}, this can be relevant in connection with $PT$-quantum mechanics and its relatives, \cite{bagbook,ben,mosta}, where the Hamiltonian of a given system is not required to be self-adjoint, but still satisfies some special requirement. For instance, the Hamiltonian could be $PT$-symmetric, $P$ and $T$ being respectively the space parity and the time reversal operators. This {\em extended} quantum mechanics has been proved to be quite relevant in the analysis of gain-loss systems, \cite{gainloss}, from a physical point of view, and from a mathematical side because of the many (and sometimes unexpected) difficulties which arise when going from self-adjoint to non self-adjoint Hamiltonians. In particular, the role of biorthogonal sets of vectors \cite{baginbagbook}, unbounded metric operators \cite{camillo,petr1} and pseudo-spectra \cite{petr2} have been widely studied in this perspective.

The paper is organized as follows: in the next section we introduce the mathematical structure needed for the analysis of our tridiagonal Hamiltonians. Then we discuss their factorization, and we use the operators introduced in this procedure to define the Susy partner of the original Hamiltonian. Of course, since this Hamiltonian $H$ is not self-adjoint, in general, we also discuss the role of $H^\dagger$ and of its Susy partner. Hence we deal with four related Hamiltonians. Among other things, we discuss the consequences of the diagonalization of $H$, showing that three-terms relations can be deduced also in this more general settings. Section \ref{sect3} is devoted to examples, which are treated in many details. In Section \ref{sect4} we consider other kind of tridiagonal matrices, and we discuss their connections with bi-squeezed states of the type originally introduced in \cite{BGS2018}. Section \ref{sect5} contains our conclusions.

\section{The functional settings}\label{sect2}

Let $\Hil$ be an Hilbert space with scalar product $\left<.,.\right>$ and related norm $\|.\|$, and let $\F_\varphi=\{\varphi_n,\, n=0,1,2,3,\ldots\}$ and $\F_\psi=\{\psi_n,\, n=0,1,2,3,\ldots\}$ be two biorthogonal sets of vectors in $\Hil$: $\left<\varphi_n,\psi_k\right>=\delta_{n,k}$. We are assuming here that $\Hil$ is infinite-dimensional, except when stated, and separable. Otherwise, if $dim(\Hil)<\infty$, the treatment simplifies significantly, from the mathematical point of view, mainly because all the operators necessary for us are bounded. In what follows $\F_\varphi$ and $\F_\psi$ will be required to be either $\D$-quasi bases or, much stronger requirement, Riesz bases, \cite{baginbagbook}. For readers' convenience we recall that $\F_\varphi$ and $\F_\psi$ are $\D$-quasi bases if  $\D$ is some dense subspace of $\Hil$, and if, for all $f,g\in\D$,
\bea
\left<f,g\right>=\sum_{n=0}^{\infty}\left<f,\varphi_n\right>\left<\psi_n,g\right>=\sum_{n=0}^{\infty}\left<f,\psi_n\right>\left<\varphi_n,g\right>.\label{r1}
\ena
Quite often $\varphi_n$ and $\psi_n$ also belongs to $\D$. This will be assumed in this paper, as useful working assumption.
$\F_\varphi$ and $\F_\psi$ are (biorthogonal) Riesz bases if an orthonormal (o.n.) basis $\F_e=\{e_n,\, n=0,1,2,3,\ldots\}$ exists in $\Hil$, together with a bounded operator $R$ with bounded inverse, such that $\varphi_n=Re_n$ and $\psi_n=(R^{-1})^\dagger e_n$. In this paper we will always assume that $\D$ is stable under the action of $R$, $R^\dagger$ and their inverse. We also assume that $e_n\in\D$ for all $n$, so that $\varphi_n,\psi_n\in\D$ automatically. We refer to \cite{baginbagbook,BGST2017} for examples when these assumptions are satisfied. We observe that if $\F_\varphi$ and $\F_\psi$ are (biorthogonal) Riesz bases, they are $\D$-quasi bases. The opposite implication is false:  $\D$-quasi bases are, in general, not Riesz bases. Also, they are often not even bases, \cite{baginbagbook}. 

Let now $H$ be an operator, not necessarily bounded or self-adjoint, such that $D(H)\supseteq \D$. Hence $H$ is densely defined. In what follows it will be useful to assume also that $D(H^\dagger)\supseteq \D$.

\begin{defn}
	\label{defn1}
	$H$ is called $(\varphi,\psi)$-tridiagonal if three sequences of complex numbers exist, $\{b_n\}$, $\{a_n\}$ and $\{b_n'\}$, such that
	\be
	\left<\psi_n,H\varphi_m\right>=b_n\delta_{n,m+1}+a_n\delta_{n,m}+b_n'\delta_{n,m-1},
	\label{21}\en
	for all $n,m=0,1,2,3,\ldots$. Furthermore, $H$ is called $e$-tridiagonal if $H$ is $(e,e)$-tridiagonal.
\end{defn}
This definition extends that given in \cite{yam} in two ways: first of all, $H$ is not required to be self-adjoint. For this reason no relation is assumed, in general, between $\{b_n\}$ and $\{b_n'\}$. Also, $a_n$ could be complex or not. Secondly, we are replacing a single basis with two biorthogonal sets, $\F_\varphi$ and $\F_\psi$, none of which is even necessarily a basis. However, as often explicitly checked in concrete examples involving $\D$-quasi bases, \cite{baginbagbook}, both $\F_\varphi$ and $\F_\psi$ are assumed to be complete in $\Hil$: the only vector $f\in\Hil$ which is orthogonal to all the $\varphi_n$'s, or to all the $\psi_n$'s, is $f=0$.

\begin{lemma}\label{lemma1}
	$H$ is  $(\varphi,\psi)$-tridiagonal if and only if $H^\dagger$ is  $(\psi,\varphi)$-tridiagonal. Moreover, if $H$ leaves $\D$ stable and if $\F_\varphi$ and $\F_\psi$ are Riesz bases, then $H$ is  $(\varphi,\psi)$-tridiagonal if and only if $H_0:=R^{-1}HR$ is  $(e,e)$-tridiagonal.
\end{lemma}

The proof is easy and will not be given here. We only want to stress that $D(H_0)\supseteq \D$ and that $\D$ is stable also under the action of $H_0$.

Now, from \eqref{21} it follows that
\be
H\varphi_m=b_{m-1}'\varphi_{m-1}+a_m\varphi_m+b_{m+1}\varphi_{m+1}.
\label{22}\en
In fact, using the biorthogonality between $\F_\psi$ and $\F_\varphi$, we can rewrite equation \eqref{21} as
$$
\left<\psi_n,\left(H\varphi_m-b_{n}'\varphi_{m-1}+a_n\varphi_m+b_{n}\varphi_{m+1}\right)\right>=0,
$$
which must be satisfied for all $n$. Now, since  the set $\F_\psi$ is complete, (\ref{22}) follows. Notice that $b'_{-1}=0$ here and in the following.  In a similar way, recalling that $\psi_n\in D(H^\dagger)$ and that $\left<\psi_n,H\varphi_m\right>=\left<H^\dagger\psi_n,\varphi_m\right>$, from (\ref{21}) and from the completeness of $\F_\varphi$ we find that
\be
H^\dagger\psi_m=\overline{b_{m}'}\,\psi_{m+1}+\overline{a_m}\,\psi_m+\overline{b_{m}}\,\psi_{m-1}.
\label{23}\en
Among other things, this formula shows that $\overline{b_{0}}=b_0=0$. Also,
formulas (\ref{22}) and (\ref{23}) show that $\varphi_m$ is not an eigenstate of $H$, and that $\psi_m$ is not an eigenstate of $H^\dagger$, except if $b_m=b_m'=0$ for all $m$. Clearly, when this happens,  $H$ is diagonal, rather than tridiagonal. Now, if we are under the assumptions of Lemma \ref{lemma1}, (\ref{22}) and (\ref{23}) produce, for $H_0$, the following equalities:
\be
H_0e_m=b_{m-1}'e_{m-1}+a_me_m+b_{m+1}e_{m+1}, \quad H_0^\dagger e_m=\overline{b_{m}'}\,e_{m+1}+\overline{a_m}\,e_m+\overline{b_{m}}\,e_{m-1}.
\label{24}\en

\begin{lemma}\label{lemma2} Let us assume that  $H$ leaves $\D$ stable and that $\F_\varphi$ and $\F_\psi$ are Riesz bases. If $H_0=H_0^\dagger$, then $a_n\in\mathbb{R}$ and $b_m=\overline{b_{m-1}'}$ for all $m\geq0$. Viceversa, if $a_m\in\mathbb{R}$ and $b_m=\overline{b_{m-1}'}$, then $\left<f,H_0g\right>=\left<H_0f,g\right>$, for all $f,g\in\E$, where $\E$ is the linear span of the $e_n$'s.
\end{lemma}

The proof is a simple consequence of formula (\ref{24}). In particular, $b_m=\overline{b_{m-1}'}$ is automatically satisfied for $m=0$, since, as we have already noticed, $b_0=b'_{-1}=0$.  Of course, $\E$ is dense in $\Hil$, since $\F_e$ is an o.n. set of vectors in the dense set $\D$. Hence $\F_e$ is an o.n. basis for $\Hil$. Notice that this Lemma shows that also $H_0$ can be non self-adjoint. This is often not the case in PT-quantum mechanics, \cite{ben}, or for pseudo-hermitian operators, \cite{mosta}, where non self-adjoint Hamiltonians are shown to be similar to self-adjoint ones, and the similarity is implemented exactly as above, in $H_0:=R^{-1}HR$. But this is not what happens, in general, in this paper.

\subsection{Factorization}
\label{Factorization}
Following \cite{yam}, we now discuss when and how $H$ can be factorized, and we use this factorization to introduce two more Hamiltonians, the supersymmetric versions of $H$ and $H^\dagger$. 

Let us first introduce an operator $X$ on $\Lc_\varphi=l.s.\{\varphi_n\}$, the linear span of the $\varphi_n$'s. Of course, this set is dense in $\Hil$  if $\F_\varphi$ is complete in $\Hil$, \cite{baginbagbook}. We put
\be
X\varphi_n=c_n\varphi_n+d_n\varphi_{n-1}
\label{25}
\en
It is clear that $X$ is not a lowering operator for $\F_\varphi$, if $c_n\neq0$. Completeness of $\F_\varphi$, and its biorthogonality with $\F_\psi$, allows us to deduce that $X^\dagger\psi_n=\overline{c_n}\,\psi_n+\overline{d_{n+1}}\,\psi_{n+1}$, which is a raising operator for $\F_\psi$ only if $c_n=0$ for all $n$. Similarly, we can introduce an operator $Y$ on the linear span of the $\psi_n$'s, $\Lc_\psi$, as in (\ref{25}):
\be
Y\psi_n=c_n'\psi_n+d_n'\psi_{n-1},
\label{26}
\en
whose adjoint, $Y^\dagger$, acts on $\varphi_n$ as follows: $Y^\dagger\varphi_n=\overline{c_n'}\,\varphi_n+\overline{d_{n+1}'}\,\varphi_{n+1}$. Again, $Y$ and $Y^\dagger$ are not ladder operators, except if $c_n'=0$. Also, we require that $d_0=d_0'=0$, in order to avoid the appearance of $\varphi_{-1}$ or $\psi_{-1}$ in the two formulas above.  Now, the following can be easily checked: $H\varphi_n=Y^\dagger X\varphi_n$ if the following relations are true:
\be
a_n=c_n\,\overline{c_n'}+d_n\,\overline{d_n'}, \qquad b_{n-1}'=d_n\,\overline{c_{n-1}'}, \qquad b_{n+1}=c_n\,\overline{d_{n+1}'}.
\label{27}\en
Under the same conditions we also deduce the following equality:  $H^\dagger\psi_n=X^\dagger Y\psi_n$, which shows, not surprisingly, that also $H^\dagger$ can be factorized in terms of the same operators. In the following, to simplify the notation, we will often write $H=Y^\dagger X$ and $H^\dagger=X^\dagger Y$. The operators $X$ and $Y^\dagger$ satisfy the following commutation relation
\be
[X,Y^\dagger]\varphi_n=\left(d_{n+1}\overline{d_{n+1}'}-d_{n}\overline{d_{n}'}\right)\varphi_n+d_n\left(\overline{c_{n}'}-\overline{c_{n-1}'}\right)\varphi_{n-1}+\left(c_{n+1}-c_n\right)\overline{d_{n+1}'}\varphi_{n+1}.
\label{28}\en
Notice that, in particular, if $X$ and $Y$ are ladder operators (so that, $c_n=c_n'=0$), then this formula simplifies and returns $[X,Y^\dagger]\varphi_n=\left(d_{n+1}\overline{d_{n+1}'}-d_{n}\overline{d_{n}'}\right)\varphi_n$, which becomes the standard pseudo-bosonic commutation relation, \cite{baginbagbook,bag2,bag3,bag4}, if $d_n=d_n'=\sqrt{n}$: $[X,Y^\dagger]\varphi_n=\varphi_n$, for all $n$.

\begin{remark}

It is interesting to notice that, when $c_n=0$, even if $d_n\neq\sqrt{n}$, it is always possible to define new vectors, $\hat{\varphi}_n$, satisfying $X\hat{\varphi}_n=\sqrt{n}\hat\varphi_{n-1}$. 
It is enough to put $\hat\varphi_n=\frac{\sqrt{n!}}{d_n!}\varphi_n$, $n=0,1,2,3,\ldots$, where $d_0!=1$ and $d_n!=d_1d_2\cdots d_n$, $n\geq1$. Analogously, if $c_n'=0$ and $d_n'\neq\sqrt{n}$, it is again possible to define the new vectors $\hat\psi_n=\frac{\sqrt{n!}}{d_n'!}\psi_n$, satisfying $Y\hat{\psi_n}=\sqrt{n}\hat\psi_{n-1}$. 
Of course, this change of normalization of the vectors have consequences in formula (\ref{27}), and in the computation of $$\left<\hat\varphi_n,\hat\psi_m\right>=\frac{n!}{d_n!\,d_n'!}\delta_{n,m}.$$ In general, these two families are still biortogonal, but no longer biorthonormal. 

\end{remark}

\begin{remark}
Even if, in general, $X$ and $Y^\dag$ are not pseudo-bosonic operators, we can still consider linear combinations of them, $C:=\alpha X+\beta Y^\dagger$, $D:=\gamma X+\delta Y^\dagger$, and look for conditions on the coefficients such that $[C,D]\varphi_n=\varphi_n$, $\forall n\geq0$. In particular, if $\alpha\delta\neq\beta\gamma$, we have
\beano
[C,D]\varphi_n=(\alpha\delta-\beta\gamma)\Big( (d_{n+1}\overline{d'_{n+1}}-d_n\overline{d'_{n}}) \varphi_n+ (d_{n}\overline{c'_{n}}-d_{n}\overline{c'_{n-1}}) \varphi_{n-1}+\\
 (c_{n+1}\overline{d'_{n+1}}-c_{n}\overline{d'_{n+1}})\varphi_{n+1}\Big),
\enano
which reduces to $[C,D]\varphi_n=\varphi_n$ by fixing $c_n=c_0$, ${c_n'}={c_0'}$, $\forall n>0$ and
$d_n\overline{d'_n}=\frac{n}{\alpha\delta-\beta\gamma},\forall n\geq0$.
Consequently, we also have $[D^\dag,C^\dag]\psi_n=\psi_n,\quad \forall n\geq0$. 
We observe that $H=Y^\dagger X$ can be written in terms of the operators $C,D$ as
$$
H=-\frac{1}{(\alpha\delta-\beta\gamma)^2}\left(\delta\gamma C^2+\alpha\beta D^2-(\alpha\delta+\beta\gamma)CD +\alpha\delta\1\right).
$$
\end{remark}

Having factorized $H$ and $H^\dagger$, it is natural to consider now their Susy partners $H_{susy}=XY^\dagger$ and $H_{susy}^\dagger=YX^\dagger$. Using formula (\ref{25}) and $Y^\dagger\varphi_n=\overline{c_n'}\,\varphi_n+\overline{d_{n+1}'}\,\varphi_{n+1}$ we deduce that
\be
H_{susy}\varphi_n=B_{n-1}'\varphi_{n-1}+A_n\varphi_n+B_{n+1}\varphi_{n+1},
\label{29}\en
where
\be
A_n=c_n\,\overline{c_n'}+d_{n+1}\,\overline{d_{n+1}'}, \qquad B_{n-1}'=d_n\,\overline{c_{n}'}, \qquad B_{n+1}=c_{n+1}\,\overline{d_{n+1}'}.
\label{210}\en
Of course, (\ref{29}) implies that $H_{susy}$ is $(\varphi,\psi)$-tridiagonal:
$$
\left<\psi_n,H_{susy}\varphi_m\right>=B_n\delta_{n,m+1}+A_n\delta_{n,m}+B_n'\delta_{n,m-1},
$$
which coincides with (\ref{21}), with $(a_n,b_n,b_n')$ replaced by $(A_n,B_n,B_n')$. Hence, Lemma \ref{lemma1} implies that $H_{susy}^\dagger$ is $(\psi,\varphi)$-tridiagonal, and we can easily check that
\be
H_{susy}^\dagger\psi_n=\overline{B_{n}}\psi_{n-1}+\overline{A_n}\psi_n+\overline{B_{n}'}\psi_{n+1},
\label{211}\en
which coincides with formula (\ref{23}) with the above replacement.

\begin{remark} If $X$ and $Y$ are lowering operators, we have $c_n=c_n'=0$, and we find $a_n=d_n\overline{d_n'}$, $A_n=d_{n+1}\overline{d_{n+1}'}=a_{n+1}$, $b_n=b_n'=B_n=B_n'=0$. Hence
	$$
	H\varphi_n=a_n\varphi_n, \quad H^\dagger\psi_n=\overline{a_n}\psi_n, \quad 	H_{susy}\varphi_n=a_{n+1}\varphi_n, \quad H_{susy}^\dagger\psi_n=\overline{a_{n+1}}\psi_n,
	$$
as expected. In this case, $\F_\varphi$ and $\F_\psi$ are eigenstates of $H$ and $H_{susy}$, and of $H^\dagger$ and $H_{susy}^\dagger$, respectively.

\end{remark}

\subsection{Diagonalization of the Hamiltonians and consequences}
\label{secDiag}

As we have already noticed, if $H$ is  $(\varphi,\psi)$-tridiagonal, then $\F_\varphi$ is not a set of eigenstates of $H$. However, we can use its vectors to look for these eigenstates, at least if $\F_\varphi$ is a basis for $\Hil$, which is what we will assume here. This implies  that its  biorthogonal set $\F_\psi$ is a basis as well, \cite{chri}.

Let $\Phi_n$ be an eigenstate of $H$, with eigenvalue $E_n$:
\be
H\Phi_n=E_n\Phi_n.\label{212}\en 
Of course, in general,  $E_n$ is also unknown. We expand $\Phi_n$ in terms of $\F_\varphi$, and we use its biorthogonality with $\F_\psi$. Hence we have
\be
\Phi_n=\sum_kc_k^{(n)}\,\varphi_k, \qquad c_k^{(n)}=\left<\psi_k,\Phi_n\right>.
\label{213}\en
Now, assuming that $H\sum_kc_k^{(n)}\,\varphi_k=\sum_kc_k^{(n)}\,H\varphi_k$, which is true, for instance, if $H$ is bounded or under some closability condition on $H$, and using (\ref{22}) and the biorthogonalities of $\F_\psi$ and $\F_\varphi$, we deduce the following relation between the coefficients:
\be
 c_l^{(n)}E_n=c_{l-1}^{(n)}b_l+c_l^{(n)}a_l+c_{l+1}^{(n)}b_l',
 \label{214}\en
where $c_{-1}^{(n)}=0$. In complete analogy we can look for eigenstates of $H^\dagger$ using $\F_\psi$: let $\eta_n$ be the eigenstate of $H^\dagger$ corresponding to the eigenvalue $\overline{E_n}$:
$$
H^\dag\eta_n=\overline{E_n}\eta_n.$$
We expand $\eta_n$ in terms of $\F_\psi$:
\be
\eta_n=\sum_kd_k^{(n)}\,\psi_k, \qquad d_k^{(n)}=\left<\varphi_k,\eta_n\right>.
\label{215}\en
Now, if $H^\dagger\sum_kd_k^{(n)}\,\psi_k=\sum_kd_k^{(n)}\,H^\dagger\psi_k$, we deduce the following relation, quite similar to that in (\ref{214}):
\be
\overline{d_l^{(n)}}E_n=\overline{d_{l-1}^{(n)}}b'_{l-1}+\overline{d_l^{(n)}}a_l+\overline{d_{l+1}^{(n)}}b_{l+1},
\label{216}\en
where, obviously, we have set $d_{-1}^{(n)}=0$.
A comparison between this formula and (\ref{214}) shows that, if $b_l=b_{l-1}'$, once $c_l^{(n)}$ is computed, then $d_l^{(n)}$ can be easily deduced by taking $d_l^{(n)}=\overline{c_l^{(n)}}$.

\begin{remark}
	Notice that,  if $b_l=b_{l-1}'$, formula (\ref{21}) becomes $\left<\psi_n,H\varphi_m\right>=b_n\delta_{n,m+1}+a_n\delta_{n,m}+b_{n+1}\delta_{n,m-1}$, which is, if $\varphi_n=\psi_n$, the starting point of the analysis proposed in \cite{yam}.
\end{remark}

The coefficients $c_l^{(n)}$ and $d_l^{(n)}$ satisfy some summation formulas which are deduced in the following Proposition.

\begin{prop}\label{prop4} The coefficients $c_l^{(n)}$ and $d_l^{(n)}$ satisfy the equation
	\be
	\sum_k \overline{c_k^{(n)}}\,d_k^{(m)}=\left<\Phi_n,\eta_m\right>=\delta_{n,m},
	\label{217}\en
	where the last equality holds if each eigenvalue of $H$ has multiplicity one and if the normalizations of $\Phi_n$ and $\eta_n$ are chosen in such a way that $\left<\Phi_n,\eta_n\right>=1$.   
	
	Also, if $\F_\Phi=\{\Phi_n\}$ and $\F_\eta=\{\eta_n\}$ are $\D$-quasi bases, then
	\be
	\sum_n \overline{c_k^{(n)}}\,d_l^{(n)}=\delta_{k,l},
	\label{218}\en

\end{prop}

\begin{proof}
First of all, using the resolution of the identity in $\D$ given by \eqref{r1} we have
$$
\sum_k \overline{c_k^{(n)}}\,d_k^{(m)}=\sum_k \left<\Phi_n,\psi_k\right>\left<\varphi_k,\eta_m\right>=\left<\Phi_n,\eta_m\right>.
$$
The fact that $\left<\Phi_n,\eta_m\right>=0$ if $n\neq m$, at least if the multiplicity of $E_n$ is one, is well known.

Equation (\ref{218}) can be proved as follows:
$$
	\sum_n \overline{c_k^{(n)}}\,d_l^{(n)}=	\sum_n \left<\varphi_l,\eta_n\right>\left<\Phi_n,\psi_k\right>= \left<\varphi_l,\psi_k\right>=\delta_{k,l},
$$
where we have used the hypothesis that  $\F_\Phi=\{\Phi_n\}$ and $\F_\eta=\{\eta_n\}$ are $\D$-quasi bases and that $\varphi_l,\psi_k\in\D$. The last equality follows from the biorthogonality of $\F_\varphi$ and $\F_\psi$.

\end{proof}

Defining next the following quantities
\be
p_l^{(n)}=\frac{c_l^{(n)}}{c_0^{(n)}}, \qquad \qquad q_l^{(n)}=\frac{d_l^{(n)}}{d_0^{(n)}},
\label{219}\en
we observe that
\be
p_{-1}^{(n)}=q_{-1}^{(n)}=0, \qquad p_{0}^{(n)}=q_{0}^{(n)}=1, \qquad \forall n\geq0.
\label{220}\en

Formulas (\ref{214}) and (\ref{216}) can be rewritten as the following recurrence equations:
\be
p_{l+1}^{(n)}=\frac{1}{b_l'}\left(p_l^{(n)}(E_n-a_l)-p_{l-1}^{(n)}b_l\right)
\label{221}\en
and
\be
\overline{q_{l+1}^{(n)}}=\frac{1}{b_{l+1}}\left(\overline{q_l^{(n)}}(E_n-a_l)-\overline{q_{l-1}^{(n)}}b_{l-1}'\right)
\label{222}\en
which produce, in principle, the sequences $\{p_l^{(n)}\}$ and $\{q_l^{(n)}\}$, and $\{c_l^{(n)}\}$ and $\{d_l^{(n)}\}$ from (\ref{219}) as a consequence,  using (\ref{220}). Of course, $E_n$ must be known in order to compute explicitly  these coefficients. This is what happens in some situations, as the examples in the next section show.

We conclude this section adapting these results, and formulas (\ref{221}) and (\ref{222}) in particular, to the Susy partners of $H$ and $H^\dagger$. We recall that they are both tridiagonal. In particular, $H_{susy}$ is $(\varphi,\psi)$-tridiagonal, and $H_{susy}^\dagger$ is $(\psi,\varphi)$-tridiagonal. Also, we have already noticed that one can go from $(H,H^\dagger)$ to $(H_{susy},H_{susy}^\dagger)$ simply replacing $(a_n,b_n,b_n')$ with $(A_n,B_n,B_n')$. Hence, starting with the following eigenvalue equations,
\be
H_{susy}\tilde{\Phi}_n=\E_n\tilde{\Phi}_n, \qquad \qquad H_{susy}^\dagger\tilde{\eta}_n=\overline{\E_n}\,\tilde{\eta}_n,
\label{223}\en
and expanding $\tilde\Phi_n$ and $\tilde{\eta_n}$ as follows,
$$
\tilde\Phi_n=\sum_k\tilde c_k^{(n)}\,\varphi_k, \quad \tilde\eta_n=\sum_k\tilde d_k^{(n)}\,\psi_k, \qquad \tilde c_k^{(n)}=\left<\psi_k,\tilde\Phi_n\right>,\quad  \tilde d_k^{(n)}=\left<\varphi_k,\tilde\eta_n\right>,
$$
the following counterparts of (\ref{221}) and (\ref{222}) can be found:
\be
P_{l+1}^{(n)}=\frac{1}{B_l'}\left(P_l^{(n)}(\E_n-A_l)-P_{l-1}^{(n)}B_l\right)
\label{224}\en
and
\be
\overline{Q_{l+1}^{(n)}}=\frac{1}{B_{l+1}}\left(\overline{Q_l^{(n)}}(\E_n-A_l)-\overline{Q_{l-1}^{(n)}}B_{l-1}'\right).
\label{225}\en
Here we have introduced the {\em normalized} coefficients \be
P_l^{(n)}=\frac{\tilde c_l^{(n)}}{\tilde c_0^{(n)}}, \qquad \qquad Q_l^{(n)}=\frac{\tilde d_l^{(n)}}{\tilde d_0^{(n)}},
\label{226}\en
which obey, in particular,
\be
P_{-1}^{(n)}=Q_{-1}^{(n)}=0, \qquad P_{0}^{(n)}=Q_{0}^{(n)}=1, \qquad \forall n\geq0.
\label{227}\en
Of course, $\tilde c_k^{(n)}$ and $\tilde d_k^{(n)}$ satisfy the analogous of Proposition \ref{prop4}. In particular, for instance, if $\F_{\tilde\Phi}=\{\tilde\Phi_n\}$ and $\F_{\tilde\eta}=\{\tilde\eta_n\}$ are $\D$-quasi bases, then
$
\sum_n \overline{\tilde c_k^{(n)}}\,\tilde d_l^{(n)}=\delta_{k,l}.$

\section{Examples}\label{sect3}

This section is devoted to the analysis of some examples of our general framework. In particular, in Section \ref{sect3.a} we propose a rather general method to produce general non self-adjoint tridiagonal matrices. In Section \ref{sect3.b} we analyse in all details a shifted harmonic oscillator, with particular attention to the three terms relations previously introduced.

\subsection{A shifted quantum well}\label{sect3.a}

Let $H_0=p^2+V(x)$, where $p=-i\frac{d}{dx}$ is the momentum operator and $V(x)$ is the potential which is zero for $x\in[0,\pi]$, and infinite outside this region. $H_0$ is therefore the self-adjoint Hamiltonian of a particle of mass $m=\frac{1}{2}$ in an infinitely deep square-well potential. It is well known that 
\be
H_0\,e_n(x)=E_ne_n(x), \qquad E_n=(n+1)^2,\qquad e_n(x)=\sqrt{\frac{2}{\pi}}\,\sin((n+1)x),
\label{sqw1}\en
where $x\in[0,\pi]$ and $n=0,1,2,3,\ldots$. In \cite{Dong} it is shown how $H_0$, as well as the Hamiltonians of many other physical systems, can be factorized. First we introduce the number operator $\hat N$ defined on the vectors $e_n(x)$, which all together form an o.n. basis for $\Hil=\Lc^2(0,\pi)$:  $\hat Ne_n=ne_n$, $n\geq0$. Of course $\hat N$ is not bounded and it is not invertible. However, $\hat N+\1$ is invertible, and $(\hat N+\1)^{-1}$ is bounded. Following \cite{Dong} we define the following operators:
$$
\hat M_+=\cos(x)(\hat N+\1)+\sin(x)\frac{d}{dx}, \quad \hat M_-=\left(\cos(x)(\hat N+\1)-\sin(x)\frac{d}{dx}\right)(\hat N+\1)^{-1}\hat N.
$$
They are ladder operators since they satisfy
$$
\hat M_+e_n=(n+1)e_{n+1}, \qquad \hat M_-e_n=ne_{n-1},
$$
where we put $e_{-1}=0$. Hence it is possible to see that $H_0e_n=\hat M_-\hat M_+e_n$: Furthermore, we cal also check that
$$
\hat M_+^\dagger e_n=\hat M_-e_n, \qquad \hat M_-^\dagger e_n=\hat M_+e_n,\qquad  [\hat M_-,\hat M_+]e_n=(2\hat N+\1)e_n
$$
for all $n$. Now, let us consider the following shifted version of the ladder operators $\hat M_\pm$: $B=\hat M_++\alpha\1$, $A=\hat M_-+\beta\1$, $\alpha, \beta\in\mathbb{C}$, and the related {\em shifted Hamiltonian}  $h=BA$. It is easy to check that $h$ is $(e,e)$-tridiagonal:
\be
he_n=\alpha ne_{n-1}+(n^2+\alpha\beta)e_n+\beta (n+1)e_{n+1},
\label{sqw2}\en
which coincides with (\ref{22}) taking $b_{n-1}'=\alpha n$, $b_{n+1}=\beta(n+1)$ and $\alpha_n=n^2+\alpha\beta$. Now, since $Ae_n=\beta e_n+ne_{n-1}$ and $B=\alpha e_n+(n+1)e_{n+1}$, the coefficients in (\ref{25}) and (\ref{26}) are $c_n=\beta$, $c_n'=\overline{\alpha}$, $d_n=d_n'=n$ and the identities in (\ref{210}) are satisfied.

As for the other Hamiltonians connected to $h$, it is easy to check that for $h^\dagger$, which is clearly $(e,e)$-tridiagonal in view of Lemma \ref{lemma1} (as an explicit computation also shows), coincides with $h$ but with $\alpha$ replaced by $\overline{\beta}$ and viceversa. As for their Susy partners, we have, for instance
$$
h_{susy}=AB=[A,B]+h= h+(2\hat N+\1),
$$
since $[A,B]=[\hat M_-+\beta\1,\hat M_++\alpha\1]=[\hat M_-,\hat M_+]=(2\hat N+\1)$. It follows that
$$
h_{susy}e_n=\alpha ne_{n-1}+((n+1)^2+\alpha\beta)e_n+\beta (n+1)e_{n+1},
$$
which shows that $A_n=a_{n+1}$, $B_n=b_n$ and $B_n'=b_n'$.

\vspace{2mm}

{\bf Remark:--} It is clear that the same approach can be extended to all systems whose self-adjoint Hamiltonian can be factorized in terms of ladder operators, as those included in \cite{Dong}. Once we have an $\tilde H_0=\tilde H_0^\dagger=Q^\dagger Q$, with eigenstates $f_n$ and eigenvalues $\E_n$, $\tilde H_0f_n=\E_nf_n$, shifting $Q$ and $Q^\dagger$ with two different complex quantities, $Q\rightarrow Q+\beta\1$ and $Q^\dagger\rightarrow Q^\dagger+\alpha\1$, with $\alpha$ possibly different from $\overline{\beta}$, the non self-adjoint operator $\tilde h=( Q^\dagger+\alpha\1)(Q+\beta\1)$ is $(f,f)$-tridiagonal, with obvious notation.  What is not easy, or possible, in general, is to make use of the recurrence relation (\ref{221}) to deduce the eigenstates of $\tilde h$, since its eigenvalues are not known a priori. In the next example and in Section \ref{sect4} we will discuss an example where this is not so, and the recurrence relations can be efficiently used to deduce the eigenvectors of the analogous of $\tilde h$.

\subsection{The shifted harmonic oscillator}\label{sect3.b}

This model has been discussed by several authors, in slightly different forms, mainly in the context of pseudo-hermitian (or PT) quantum mechanics, \cite{ben,mosta}. Some useful references are \cite{dav1,dav2,bag5,bag6,siegl,milos}.

Let $c$ be a lowering operator on $\Hil$ satisfying the canonical commutation relation $[c,c^\dagger]=\1$. Of course, this equality must be understood on a suitable dense subspace of $\Hil$, since $c$ and $c^\dagger$ are unbounded. For instance, if $c=\frac{1}{\sqrt{2}}\left(x+\frac{d}{dx}\right)$, the Hilbert space is $\Hil=\Lc^2(\mathbb{R})$ and the dense set can be identified with $\Sc(\mathbb{R})$, the set of the fast decreasing test functions. If we introduce the vacuum of $c$, that is a (normalized) vector $e_0\in\Hil$ satisfying $ce_0=0$, we can act on it with powers of $c^\dagger$: $e_n=\frac{(c^\dagger)^n}{\sqrt{n!}}\,e_0$. The resulting vectors, $\{e_n\}$, form an o.n. basis of $\Hil$, which is all made by functions of $\Sc(\mathbb{R})$ if $c$ is represented as above. These vectors are eigenstates of $H_0=c^\dagger c$: $H_0e_n=ne_n$, $n=0,1,2,\ldots$.

Let us now define $a=c+\alpha\1$ and $b=c^\dagger+\beta\1$, for some $\alpha,\beta\in\mathbb{C}$, with $\alpha\neq\overline{\beta}$. These operators are $\D$-pseudo bosonic, \cite{baginbagbook,bag5,bag6}, where, using the coordinate representation for $c$ and $c^\dagger$,  $\D$ can be identified with $\Sc(\mathbb{R})$. In particular, for instance, $[a,b]f=f$ for all $f\in\D$. If we now call $H=ba=H_0+(\alpha c^\dagger+\beta c)+\alpha\beta\1$, we find that
\be
He_n=(n+\alpha\beta)e_n+\alpha\sqrt{n+1}e_{n+1}+\beta\sqrt{n}e_{n-1},
\label{31}
\en
so that
$\left<e_n,He_m\right>=(n+\alpha\beta)\delta_{n,m}+\alpha\sqrt{n}\delta_{n,m+1}+\beta\sqrt{n+1}\delta_{n,m-1}$. We see that $H$ is $(e,e)$-tridiagonal, like $H^\dagger$. Incidentally, we also observe that $H^\dagger$ coincides with $H$, but with $(\alpha,\beta)$ replaced by $(\overline{\beta},\overline{\alpha})$. 

Now, since $ce_n=\sqrt{n}\,e_{n-1}$ and $c^\dagger e_n=\sqrt{n+1}\,e_{n+1}$, we see that $ae_n=(c+\alpha\1)e_n=\alpha e_n+\sqrt{n}\,e_{n-1} $, while $be_n=(c^\dagger+\beta\1)e_n=\beta e_n+\sqrt{n+1}\,e_{n+1} $, so that $X=a$ and $Y^\dagger=b$ only if the following identifications hold:
\be
c_n=\alpha, \quad c_n'=\overline{\beta}, \quad d_n=d_n'=\sqrt{n}. \label{32}\en Therefore, since formula (\ref{31}) implies that $b_n=\alpha\sqrt{n}$, $a_n=n+\alpha\beta$ and $b_n'=\beta\sqrt{n+1}$, the equalities in (\ref{27}) are satisfied. It is clear that, in the present example, the commutation relation in (\ref{28}) simplifies: $[X,Y^\dagger]e_n=[a,b]e_n=e_n$, for all $n=0,1,2,3,\ldots$.

As for $H_{susy}=ab$, we easily see that 
$$
H_{susy}e_n=([a,b]+H)e_n=(H+\1)e_n=(n+1+\alpha\beta)e_n+\alpha\sqrt{n+1}e_{n+1}+\beta\sqrt{n}e_{n-1},
$$
which coincides with (\ref{31}) expect that $n+\alpha\beta$ is now replaced by $n+1+\alpha\beta$. We observe that $A_n=a_{n+1}$, $B_n=b_n$ and $B_n'=b_n'$, and that 
$$
H_{susy}^\dagger e_n=(n+1+\overline{\alpha}\overline{\beta})e_n+\overline{\alpha}\sqrt{n}e_{n-1}+\overline{\beta}\sqrt{n+1}e_{n+1}.
$$
It is now interesting to discuss the role of (\ref{221}) and (\ref{222}) in this example. This is particularly simple here since we know that the eigenvalues of $H$ and $H^\dagger$ are just $E_n=n$, for all $n=0,1,2,\ldots$.

Let us first take $n=0$, and look for the ground state of $H=ba$: $H\Phi_0=0$. Such an eigenstate can be easily found, simply by looking at the vacuum of $a$. Of course, $a\Phi_0=0$ if and only if $c\Phi_0=-\alpha\Phi_0$. This means that $\Phi_0$ is (proportional to) a standard coherent state, \cite{gazeaubook,perelomov,didier,aag}, with parameter $-\alpha$:
\be
\Phi_0=N_\Phi e^{-\alpha c^\dagger+\overline{\alpha}c}e_0=N_\Phi e^{-\frac{|\alpha|^2}{2}}\sum_{k=0}^\infty\frac{(-\alpha)^k}{\sqrt{k!}}\,e_k,
\label{33}\en
where $N_\Phi$ is a normalization factor which is usually taken equal to one for standard coherent states, \cite{gazeaubook}.

In a similar way we could find the ground state of $H^\dagger$. However, the easier way to find $\eta_0$ is just to recall the above cited symmetry between $H$ and $H^\dagger$. Hence $\eta_0$ is, a part the normalization, nothing but $\Phi_0$ with $\alpha$ replaced by $\overline{\beta}$:
\be
\eta_0=N_\eta e^{-\overline{\beta} c^\dagger+\beta c}e_0=N_\eta e^{-\frac{|\beta|^2}{2}}\sum_{k=0}^\infty\frac{(-\overline{\beta})^k}{\sqrt{k!}}\,e_k.
\label{34}\en
A connection between $N_\Phi$ and $N_\eta$ can be found by requiring that $\left<\Phi_0,\eta_0\right>=1$: $N_\Phi N_\eta=e^{\frac{1}{2}(|\alpha|^2+|\beta|^2+\alpha\beta)}$.

We want to show now that the same expansions as in (\ref{33}) and (\ref{34}) can be obtained by means of (\ref{221}) and (\ref{222}). We start specializing (\ref{221})  to $n=0$ and to our particular value of the coefficients: 
$$
p_{l+1}^{(0)}=\frac{-1}{\beta\sqrt{l+1}}\left(p_l^{(0)}(l+\alpha\beta)+p_{l-1}^{(0)}\alpha\sqrt{l}\right),
$$
with, as usual, $p_{-1}^{(0)}=0$ and  $p_{0}^{(0)}=1$. It is simple now to find the general solution of this recurrence relation: $p_k^{(0)}=\frac{(-\alpha)^k}{\sqrt{k!}}$, so that $c_k^{(0)}=\frac{(-\alpha)^k}{\sqrt{k!}}\,c_0^{(0)}$, for all $k=0,1,2,\ldots$. Hence, formula (\ref{213}) produces $$\Phi_0=\sum_kc_k^{(0)}\,e_k=c_0^{(0)}\sum_k\frac{(-\alpha)^k}{\sqrt{k!}}\,e_k,$$
which coincides with (\ref{33}), upon identifying $c_0^{(0)}$ with $N_\Phi e^{-\frac{|\alpha|^2}{2}}$.

Using now (\ref{222}), in the same way we recover $\eta_0$ in (\ref{34}). This is because we find $q_k^{(0)}=\frac{(-\overline{\beta})^k}{\sqrt{k!}}$.

{Notice that, in this simple example, we can also make use of the factorization $H=Y^\dag X$ to get the same results.
In fact, the ground $\Phi_0$ can be obtained as the vacuum of $X, X\Phi_0=0$ (and similarly $\eta_0$ as the ground of $Y$). Expanding $\Phi_0$ as
$$
\Phi_0=\sum_{k\geq0} c^{(0)}_k\varphi_k,\quad c^{(0)}_k=\left<\psi_k,\Phi_0\right>, 
$$ 
and using the biorthogonality conditions between $\F_\varphi$ and $\F_\psi$ we have
$$
0=\left<\psi_k,X\Phi_0\right>=c_kc^{(0)}_k+d_{k+1}c^{(0)}_{k+1}=\alpha c^{(0)}_k+\sqrt{k+1}c^{(0)}_{k+1},
$$
and as before the solution is $c^{(0)}_{k}=\frac{(-\alpha)^k}{\sqrt{k!}}c^{(0)}_{0}$, and therefore
$p^{(0)}_{k}=\frac{(-\alpha)^k}{\sqrt{k!}}$.
}

\vspace{2mm}

We now generalize these results to the higher energetic levels, $n>0$, and  show that the eigenstates of $H$ can be completely  determined again by using relation~\eqref{221}. First, for pedagogical reason, we discuss the case $n=1$ and then we extend the results. 

The eigenstates of $H$ are given by (\cite{baginbagbook}, p.~148) 

\begin{equation}\label{fed1}
	\Phi_n = \frac{1}{\sqrt{n!}} b^n \Phi_0 
\end{equation}
where $\Phi_0$ is as in \eqref{33}. It is easy to verify that

\begin{equation}
b^n e_k = \sum_{i=0}^n \left[ \binom{n}{i} \beta^{n-i} p_k(i) \right] e_{k+i}
\end{equation}
where $p_k(i) = \sqrt{k+1}\sqrt{k+2}\ldots\sqrt{k+i} $ if $i\geq 1$ and 0 if $i=0$. Therefore

\begin{equation}
\Phi_n = \frac{1}{\sqrt{n!}} N_\Phi e^{-\frac{|\alpha|^2}{2}}\sum_{k=0}^\infty\frac{(-\alpha)^k}{\sqrt{k!}}\, \sum_{i=0}^n \left[ \binom{n}{i} \beta^{n-i} p_k(i) \right]e_{k+i},
\end{equation}

The first ``excited'' state will be

\begin{equation}\label{phi1}
\Phi_1 =N_\Phi e^{-\frac{|\alpha|^2}{2}} \sum_{k=0}^\infty\frac{(-\alpha)^k}{\sqrt{k!}}\,  \left[ k + (-\alpha)\beta \right]e_{k},
\end{equation}

This result can also be recovered by starting from the recurrence relation~\eqref{221}, which looks now as follows:

\begin{equation}
	p_{l+1}^{(1)}=\frac{1}{\beta\sqrt{l+1}}\left(p_l^{(1)}(1-(l+\alpha\beta))-p_{l-1}^{(1)}\alpha\sqrt{l}\right)
\end{equation}
with $p_{-1}^{(1)} = 0 $ and $p_{0}^{(1)} = 1 $. It is easy to show that 

\begin{equation}
	c_{l}^{(1)} = \frac{c_{0}^{(1)}}{\beta} \sum_{l=0}^\infty\frac{(-\alpha)^l}{\sqrt{l!}}\,  \left[ l + (-\alpha)\beta \right]
\end{equation}
which allows to retrieve \eqref{phi1} provided that $c_{0}^{(1)} = \beta N_\Phi e^{-\frac{|\alpha|^2}{2}}$.

For arbitrary $n>1$ it is possible to write $\Phi_n$ as

\begin{equation}
	\Phi_n = \frac{1}{\sqrt{n!}} N_\Phi e^{-\frac{|\alpha|^2}{2}}
	\sum_{k=0}^\infty\frac{(-\alpha)^k}{\sqrt{k!}}\, 
	\sum_{j=0}^n \left[ \binom{n}{j} [(-\alpha)\beta]^{n-j} j! \binom{k}{j} \right]e_{k}
\end{equation}
and show that the recurrence relation yields the same result provided that 

\begin{equation}
	c_{0}^{(n)} =  \frac{\beta^n}{\sqrt{n!}} N_\Phi e^{-\frac{|\alpha|^2}{2}}
\end{equation}

Using the symmetry between $H$ and $H^\dagger$ it is easy to see that the ``excited'' states of $H^\dagger$ are given by  (\cite{baginbagbook}, p.~148)

\begin{equation}\label{fed2}
\eta_n = \frac{1}{\sqrt{n!}} (a^\dagger)^n \eta_0 = \frac{1}{\sqrt{n!}} N_\eta e^{-\frac{|\beta|^2}{2}}
\sum_{k=0}^\infty\frac{(-\overline{\beta})^k}{\sqrt{k!}}\, 
\sum_{j=0}^n \left[ \binom{n}{j} [(-\overline{\beta})\overline{\alpha}]^{n-j} j! \binom{k}{j} \right]e_{k}
\end{equation}
and this is the same result one gets starting from the recurrence relation~\eqref{222}, a part for a normalization factor.

A similar analysis can be repeated also for the Susy partners of $H$ and $H^\dagger$. Of course, since in the present situation $H_{susy}=H+\1$ and $H_{susy}^\dagger=H^\dagger+\1$, the eigenstates in (\ref{223}) coincides with those without the tilde: $\tilde\Phi_n=\Phi_n$ and $\tilde\eta_n=\eta_n$, while the eigenvalues obey the relation $\E_n=E_n+1=n+1$. If we now adopt (\ref{224}) and (\ref{225}), with $\E_0=1$, we recover again the correct eigenstates, a part for the normalization, which must be chosen with care.

\begin{remark}

This example can be generalized by introducing a sort of  {\em double translation}. More explicitly, we can consider, as starting points, two $\D$-pseudo bosonic operators $a$ and $b$, $[a,b]f=f$ for all $f\in\D$, and the related (already) non self-adjoint Hamiltonian $H=ba$: $H\neq H^\dagger$. Its eigenvalues are $E_n=n$, $n=0,1,2,3,\ldots$, while its eigenvectors are those in (\ref{fed1}). $H^\dagger$ has the same eigenvalues of $H$, while its eigenstates are those in (\ref{fed2}). If we now introduce two complex parameters $\gamma$ and $\delta$, and two new operators $A=a+\gamma\1$ and $B=b+\delta\1$, it is clear that $[A,B]f=f$ for all $f\in\D$. Moreover, in general, $A\neq B^\dagger$. It is easy to check that $\hat H=BA$ is $(\Phi,\eta)$-tridiagonal, and therefore, see Lemma \ref{lemma1}, $\hat H^\dagger$ is $(\eta,\Phi)$-tridiagonal. What we have discussed above can be essentially repeated, with minor changes, for $\hat H$, $\hat H^\dagger$, and for their Susy-partners. In particular, if the operators $a$ and $b$ are related to two bosonic operators $c$ and $c^\dagger$ as above, $a=c+\alpha\1$ and $b=c^\dagger+\beta\1$, it is clear that $A=c$ and $B=c^\dagger$ if $\alpha=-\gamma$ and $\beta=-\delta$. In this case, $H_0=\hat H=\hat H^\dagger$. When these equalities (or one of them) are not satisfied, the same results as in this section hold true with $(\alpha,\beta)$ replaced by $(\alpha+\gamma,\beta+\delta)$.

\end{remark}

\section{Extended settings}\label{sect4}

In this section we consider a slightly different form of the Hamiltonian which is not now tridiagonal in the sense of \eqref{21}, but whose matrix elements in two biorthogonal bases can still be written as a sum of three contributions. All the hypothesis of completeness, closability and domain invariance assumed in the previous sections are maintained, if not specified differently.

\begin{defn}
	\label{defn3}
	$H$ is called $(\varphi,\psi)_h$-tridiagonal, with $h>0$, if three sequences of complex numbers exist, $\{b_n\}$, $\{a_n\}$ and $\{b_n'\}$, such that
	\be
	\left<\psi_n,H\varphi_m\right>=b_n\delta_{n,m+h}+a_n\delta_{n,m}+b_n'\delta_{n,m-h},
	\label{41}\en
	for all $n,m=0,1,2,3,\ldots$. Furthermore, $H$ is called $e_h$-tridiagonal if $H$ is $(e,e)_h$-tridiagonal.
\end{defn}

Of course, if $h=1$ we return to the situation considered in Section \ref{sect2}. Hence, to make the situation interesting, in this section we assume $h>1$

%

Using \eqref{41} and completeness of $\F_{\varphi}$ and $\F_{\psi}$, we deduce the natural extensions of
\eqref{22} and \eqref{23}:
\bea
H\varphi_m&=&b_{m-h}'\varphi_{m-h}+a_m\varphi_m+b_{m+h}\varphi_{m+h},\\
H^\dagger\psi_m&=&\overline{b_{m}'}\,\psi_{m+h}+\overline{a_m}\,\psi_m+\overline{b_{m}}\,\psi_{m-h},
\ena
with the clear conditions that  $b_j=0$ for $j<h$ and $b'_j=0$ for $j<0$.
It is straightforward to factorize $H$ and $H^\dag$ by introducing the operator $X_h$ on $\mathcal{L}_\varphi$ and $Y_h$ on $\mathcal{L}_\psi$, defined as
\bea
X_h\varphi_n=c_n\varphi_n+d_{n}\varphi_{n-h},\quad
Y_h\psi_n=c'_n\psi_n+d'_{n}\psi_{n-h},\quad \forall n\geq0,
\ena  
with  $d_j=d'_j=0,j<h$.
It can be easily checked that
$H\varphi_n=Y_h^\dagger X_h\varphi_n$ and $H^\dag\psi_n=X_h^\dagger Y_h\psi_n$ by putting
\be
a_n=c_n\,\overline{c_n'}+d_n\,\overline{d_n'}, \qquad b_{n-h}'=d_n\,\overline{c_{n-h}'}, \qquad b_{n+h}=c_n\,\overline{d_{n+h}'},
\label{e27}\en
and that in general 
\be
[X_h,Y_h^\dagger]\varphi_n=\left(d_{n+h}\overline{d_{n+h}'}-d_{n}\overline{d_{n}'}\right)\varphi_n+d_n\left(\overline{c_{n}'}-\overline{c_{n-h}'}\right)\varphi_{n-h}+\overline{d_{n+h}'}\left(c_{n+h}-c_n\right)\varphi_{n+h}.
\label{e28}\en

To find a suitable recurrence formula for the determination of the eigenstates of $H$ and $H^\dagger$, we adopt the same strategy used in Section \ref{secDiag}.
In particular, if $\Phi_n,\eta_n$ are an eigenstates of $H$ and $H^\dagger$,  $H\Phi_n=E_n\Phi_n,\quad H^\dag\eta_n=\bar{E}_n\eta_n$, and we expand $\Phi_n$ and $\eta_n$ as in \eqref{213} and \eqref{215}, we obtain the following recurrence formulas:
\bea
 c_l^{(n)}E_n&=&c_{l-h}^{(n)}b_l+c_l^{(n)}a_l+c_{l+h}^{(n)}b_l',\label{e215}\\
\overline{d_l^{(n)}}E_n&=&\overline{d_{l-h}^{(n)}}b'_{l-h}+\overline{d_l^{(n)}}a_l+\overline{d_{l+h}^{(n)}}b_{l+h},
\label{e216}\ena
with   $c_{j}^{(n)}, d_{j}^{(n)}=0$, $j<h,j\neq0$,
and the related
\bea
p_{l+h}^{(n)}&=&\frac{1}{b_l'}\left(p_l^{(n)}(E_n-a_l)-p_{l-h}^{(n)}b_l\right),\label{e221}\\
\overline{q_{l+h}^{(n)}}&=&\frac{1}{b_{l+h}}\left(\overline{q_l^{(n)}}(E_n-a_l)-\overline{q_{l-h}^{(n)}}b_{l-h}'\right),
\label{e222}\ena
where the coefficients $p_j^{(n)},q_j^{(n)}$ are defined as in \eqref{219} with $p_{j}^{(n)}=q_{j}^{(n)}=0,\quad j<h$ with the exceptions $p_{0}^{(n)}=q_{0}^{(n)}=1$.

\subsection{A squeezed Hamiltonian} 

Despite the general $(\varphi,\psi)_h$-tridiagonal settings seems to be a straightforward extension of the $(\varphi,\psi)$-tridiagonal case,
some relevant Hamiltonians in physics can be related to them, giving rise to states having interesting features.
In the following we consider an Hamiltonian from which a (bi)-squeezed state can be obtained by applying our recurrence procedure, \cite{BGS2018}.

Suppose that there exist two pseudo-bosonic operators $a,b$ satisfying the commutation rules $[a,b]=\1$ in $\D$, dense subspace of $\Hil$.
As usual, we suppose that  $\D$ is invariant under the action of $a$, $b$, and their adjoints. Following \cite{baginbagbook} we have
\bea
a\varphi_n&=&\sqrt{n}\varphi_{n-1},\qquad b\varphi_n=\sqrt{n+1}\varphi_{n+1},\label{raising1}\\
b^\dag\psi_n&=&\sqrt{n}\psi_{n-1}, \qquad a^\dag\psi_n=\sqrt{n+1}\psi_{n+1}.
\label{raising2}
\ena
Next we introduce the {\em squeezing-like operators}, labelled by the complex variable $z=re^{i\theta},\quad r>0$:
\bea
 \mathcal{\Sc} (z)f=\sum_{k\geq0}\frac{1}{k!}\left(\frac{z}{2}b^2-\frac{\overline{z}}{2}a^2\right)^{k} f,\label{squeez_S} \quad \mathcal{T} (z)f=\sum_{k\geq0}\frac{1}{k!}\left(\frac{z}{2}(a^\dag)^2-\frac{\overline{z}}{2}(b^\dagger)^2\right)^{k} f\label{squeez_T},
\ena
for all $f\in\D$, which under our assumptions converge strongly in $\D$ to  $e^{ \frac{1}{2}zb^2-\frac{1}{2}\bar{z}a^2}$ and to $e^{ \frac{1}{2}z(a^\dag)^2-\frac{1}{2}\bar{z}(b^\dag)^2}$ respectively, see \cite{BGS2018}. 
We can now introduce the operators $A=\mathcal{\Sc} (z)a \mathcal{T}^\dag(z), B=\mathcal{T} (z)b \mathcal{\Sc}^\dag(z)$ which reduces to 
\bea
A=\cosh(r)a+e^{i\theta}\sinh(r)b,\quad B=\cosh(r)b+e^{-i\theta}\sinh(r)a,
\ena
see \cite{BGS2018}.
They look like $\D$-pseudo bosonic operators too, because they satisfy $[A,B]=\1$ in $\D$. 

We now define the Hamiltonian
\bea
H=BA=\mu(z)ba+\lambda(z)a^2+\lambda(\overline z)b^2+\sinh(r)^2\1,
\ena
where $\mu(z)=\cosh(2r),\lambda(z)=e^{-i\theta}\cosh(r)\sinh(r)$.
Of course $H$ is $(\varphi,\psi)_2$ tridiagonal,
because, using the raising and lowering conditions \eqref{raising1}-\eqref{raising2}, we have
	\be
	\left<\psi_n,H\varphi_m\right>=b_n\delta_{n,m+2}+a_n\delta_{n,m}+b_n'\delta_{n,m-2},
	\label{Hs1}\en
with 
\be a_n=n\mu(z)+\sinh(r)^2,\quad b_n=\lambda(\overline z)\sqrt{n(n-1)},\quad  b'_n=\lambda(z)\sqrt{(n+1)(n+2)},
\label{coeffsquiz}\en
for all $n\geq0$.


The eigenvalues of $H$ are clearly $E_n=n$. 
Hence, the ground $\Phi_0$ of $H$ satisfies  $H\Phi_0=0$.
To find the expressions for $\Phi_0$ we expand it as 
$$
\Phi_0=\sum_{k\geq0} c^{(0)}_k\varphi_k,\quad c^{(0)}_k=\left<\psi_k,\Phi_0\right>, 
$$ 
and we can find the coefficients $c^{(0)}_k$ by means of \eqref{e215} and \eqref{e221}.
In particular, we have
\bea
p^{(0)}_{k+2}=\frac{1}{b'_k}\left(-a_kp^{(0)}_k-b_kp^{(0)}_{k-2}\right),\label{recurrp2k}
\ena
with the initial conditions $p^{(0)}_{-2}=p^{(0)}_{-1}=p^{(0)}_{1}=0$ and $p^{(0)}_{0}=1$.
This recurrence formula admits the solution
\bea
p^{(0)}_{2k}=\left(\frac{-e^{i\theta}\tanh(r)}{2}\right)^{k}\frac{\sqrt{(2k)!}}{k!},\quad p^{(0)}_{2k+1}=0, \quad\forall k\geq0,
\label{coeffp2}
\ena
so that we have
$$
\Phi_0=c^{(0)}_0\sum_{k\geq0}\left(\frac{-e^{i\theta}\tanh(r)}{2}\right)^{k}\frac{\sqrt{(2k)!}}{k!}\phi_{2k}.
$$
Looking for the ground $\eta_0$ of $H^\dag$ we obtain in a similar way
$$
\eta_0=d^{(0)}_0\sum_{k\geq0}\left(\frac{-e^{i\theta}\tanh(r)}{2}\right)^{k}\frac{\sqrt{(2k)!}}{k!}\psi_{2k}.
$$
We notice that $\Phi_0$ and $\eta_0$ are (proportional) to the bi-squeezed states, \cite{BGS2018}.
In particular, choosing a normalization in such a way that $c^{(0)}_0=d^{(0)}_0=e^{-\frac{1}{2}\log(\cosh(r))}$, we get
$\left<\eta_0,\Phi_0\right>=1$.

Of course, we can also use the factorization $H=Y^\dag X$ to recover the same results,
and recovering $\Phi_0$ as the vacuum of $X$ ($X\Phi_0=0$).
In this case the condition \eqref{e27} with \eqref{coeffsquiz} is satisfied by choosing 
$$
c_n = c'_n = \sqrt{n + 1} \sinh(r),\quad d_n = d'_n = e^{-i\theta}\sqrt{n}\cosh(r),
$$  
and to retrieve $\Phi_0$ we require that
$$\left<\psi_k,X\Phi_0\right>=c_k c^{(0)}_k + 
d_{k+2} c^{(0)}_{k+2}=0, \quad \forall k\geq0.$$
This implies that the coefficients $p^{(0)}_k$  satisfy the recurrence formula
$$
p^{(0)}_{k+2}=-e^{i\theta}\tanh(r)\sqrt{\frac{k+1}{k+2}}p^{(0)}_{k},
$$
which again is satisfied by \eqref{coeffp2}.
The advantage of using the factorization $H=Y^\dag X$ relies in the fact that we can recover an easier recurrence formula which uses a relationship between two consecutive even coefficients $p^{(0)}_j$ only, instead of using the recurrence formula \eqref{recurrp2k}, where three terms are involved.

{Of course, once we have retrieved the ground states of $H$ and $H^\dag$,
we can easily find the ground states of their Susy partners.
In fact, as we have
$$
H_{susy}=AB=\mu(z)ba+\lambda(z)a^2+\lambda(\overline z)b^2+\cosh(r)^2\1=H+\1,
$$
the ground states $\tilde\Phi_0,\tilde\eta_0$ of $H_{susy}$ and $H^\dagger_{susy}$ coincide with $\Phi_0,\eta_0$, respectively, but with eigenvalues 1. This is not very different from what we have deduced in Section \ref{sect3.a}.
}

\section{Conclusions}\label{sect5}

In this paper we have considered non self-adjoint tridiagonal Hamiltonians and their Susy partners, and discussed the possibility to factorize them using operators which may, or may not, be pseudo-bosonic. Three-terms recurrence relations have been deduced and have been used in the construction of the eigenstates of the Hamiltonians involved in our analysis. Within the framework proposed here we have considered a shifted harmonic oscillator, and  a {  shifted} infinitely deep square well. 

Furthermore, we have extended our results to $(\varphi,\psi)_h$-tridiagonal matrices, and we have shown how  this extension, if $h=2$, is connected with squeezed and bi-squeezed states.

\section*{Acknowledgements}
The authors acknowledge partial support from Palermo University. F.B. and F.G.  acknowledge partial support from G.N.F.M. of the I.N.d.A.M.
F.G. also  acknowledge partial support from MIUR.

\end{document}